\documentclass[11pt]{article}
\usepackage[margin=1in]{geometry}
\usepackage[colorlinks=true, allcolors=blue]{hyperref}

\usepackage{fixmath}
\usepackage{bm}
\usepackage{amsbsy}
\usepackage{color}
\usepackage{verbatim}
\usepackage{multirow}
\usepackage{amssymb}
\usepackage{amsthm}
\usepackage{array}
\usepackage{mathtools}
\usepackage{amsmath}
\usepackage{graphicx}
\usepackage{tikz}
\usetikzlibrary{positioning}
\usepackage{xcolor}
\usepackage{thm-restate}

\newtheorem{theorem}{Theorem}
\newtheorem{lemma}[theorem]{Lemma}
\newtheorem{corollary}[theorem]{Corollary}

\newcommand{\lra}[1]{\ensuremath{(#1)}}

\newcommand{\lrA}[1]{\ensuremath{\left(#1\right)}}

\newcommand{\lrC}[1]{\ensuremath{\left\{#1\right\}}}

\def\C{\mathcal{C}}
\def\M{\mathcal{M}}

\def\P{\mathcal{P}}

\def\E{\mathcal{E}}

\def\S{\mathcal{S}}

\let\epsilon=\varepsilon

\newcommand{\EE}[1]{\ensuremath{\mathbb{E}[#1]}}

\title{Approximation Algorithms for Packing Cycles and Paths in Complete Graphs}
\author
{
Jingyang Zhao\footnote{University of Electronic Science and Technology of China. Email: \texttt{jingyangzhao1020@gmail.com}.}
\and
Mingyu Xiao\footnote{University of Electronic Science and Technology of China.
Email: \texttt{myxiao@gmail.com}.}
}

\date{}

\begin{document}

\maketitle

\begin{abstract}
Given an edge-weighted (metric/general) complete graph with $n$ vertices, the maximum weight (metric/general) $k$-cycle/path packing problem is to find a set of $\frac{n}{k}$ vertex-disjoint $k$-cycles/paths such that the total weight is maximized. In this paper, we consider approximation algorithms. For metric $k$-cycle packing, we improve the previous approximation ratio from $3/5$ to $7/10$ for $k=5$, and from $7/8\cdot(1-1/k)^2$ for $k>5$ to $(7/8-0.125/k)(1-1/k)$ for constant odd $k>5$ and to $7/8\cdot (1-1/k+\frac{1}{k(k-1)})$ for even $k>5$. For metric $k$-path packing, we improve the approximation ratio from $7/8\cdot (1-1/k)$ to $\frac{27k^2-48k+16}{32k^2-36k-24}$ for even $10\geq k\geq 6$. For the case of $k=4$, we improve the approximation ratio from $3/4$ to $5/6$ for metric 4-cycle packing, from $2/3$ to $3/4$ for general 4-cycle packing, and from $3/4$ to $14/17$ for metric 4-path packing.

\end{abstract}

\section{Introduction}
In a graph with $n$ vertices, a \emph{$k$-cycle/path packing} is a set of $\frac{n}{k}$ vertex-disjoint $k$-cycles/paths (i.e., a simple cycle/path on $k$ different vertices) covering all vertices. For an edge-weighted complete graph, every edge has a non-negative weight. Moreover, it is called a \emph{metric} graph if the weight satisfies the triangle inequality; Otherwise, it is called a \emph{general} graph. Given a (metric/general) graph, the maximum weight (metric/general) $k$-cycle/path packing problem ($k$CP/$k$PP) is to find a $k$-cycle/path packing such that the total weight of the $k$-cycles/paths in the packing is maximized.

When $k=n$, $k$CP is the well-known maximum weight traveling salesman problem (MAX TSP). One may obtain approximation algorithms of $k$CP and $k$PP by using approximation algorithms of MAX TSP. In the following, we let $\alpha$ (resp., $\beta$) denote the current-best approximation ratio of MAX TSP on metric (resp., general) graphs. We have $\alpha=7/8$~\cite{DBLP:journals/tcs/KowalikM09} and $\beta=4/5$~\cite{DBLP:conf/ipco/DudyczMPR17}.

\subsection{Related Work}
For $k=2$, $k$CP and $k$PP are equivalent with the maximum weight perfect matching problem, which can be solved in $O(n^3)$ time~\cite{gabow1974implementation,lawler1976combinatorial}. For $k\geq3$, metric $k$CP and $k$PP become NP-hard~\cite{KirkpatrickH78}, and general $k$CP and $k$PP become APX-hard even on $\{0,1\}$-weighted graphs (i.e., a complete graph with edge weights 0 and 1)~\cite{manthey2008approximating}. There is a large number of results on approximation algorithms.

\textbf{General $k$CP.}
For $k=3$, Hassin and Rubinstein~\cite{hassin2006approximation,DBLP:journals/dam/HassinR06a} proposed a randomized $(0.518-\epsilon)$-approximation algorithm, Chen et al.~\cite{DBLP:journals/dam/ChenTW09,DBLP:journals/dam/ChenTW10} proposed an improved randomized $(0.523-\epsilon)$-approximation algorithm, and Van Zuylen~\cite{DBLP:journals/dam/Zuylen13} proposed a deterministic algorithm with the same approximation ratio. For lager $k$, Li and Yu~\cite{li2023cyclepack} proposed a $2/3$-approximation algorithm for $k=4$ and a $\beta\cdot(1-1/k)^2$-approximation algorithm for $k\geq5$. On $\{0,1\}$-weighted graphs, Bar-Noy et al.~\cite{DBLP:journals/dam/Bar-NoyPRV18} gave a $3/5$-approximation algorithm for $k=3$, and Berman and Karpinski~\cite{DBLP:conf/soda/BermanK06} gave a $6/7$-approximation algorithm for $k=n$.
%Note that Berman and Karpinski~\cite{DBLP:conf/soda/BermanK06} gave a $6/7$-approximation algorithm for the \emph{Maximum Path Cover Problem}, which seeks a set of node disjoint paths such that the number of edges in all the paths is maximized. Their algorithm could be used to obtain a $(6/7-O(1/n))$-approximation algorithm for general $n$CP and a $6/7$-approximation algorithm for $n$PP on $\{0,1\}$-weighted graphs.

\textbf{Metric $k$CP.}
%For $k=3$, Hassin et al.~\cite{hassin1997approximation1} firstly gave a deterministic $2/3$-approximation algorithm and Chen et al.~\cite{DBLP:journals/jco/ChenCLWZ21} proposed a randomized $(0.66768-\epsilon)$-approximation algorithm.
For $k=3$, Hassin et al.~\cite{hassin1997approximation1} firstly gave a deterministic $2/3$-approximation algorithm, Chen et al.~\cite{DBLP:journals/jco/ChenCLWZ21} proposed a randomized $(0.66768-\epsilon)$-approximation algorithm, and Zhao and Xiao~\cite{DBLP:journals/corr/abs-2402-08216} proposed a deterministic $(0.66835-\epsilon)$-approximation algorithm.
For lager $k$, Li and Yu~\cite{li2023cyclepack} proposed a $3/4$-approximation algorithm for $k=4$, a $3/5$-approximation algorithm for $k=5$, and an $\alpha\cdot(1-1/k)^2$-approximation algorithm for $k\geq6$.

\textbf{General $k$PP.}
For $k=3$, Hassin and Rubinstein~\cite{hassin2006approximation} proposed a randomized $(0.5223-\epsilon)$-approximation algorithm, Chen et al.~\cite{tanahashi2010deterministic} proposed a deterministic $(0.5265-\epsilon)$-approximation algorithm, and Bar-Noy et al.~\cite{DBLP:journals/dam/Bar-NoyPRV18} proposed an improved $7/12$-approximation algorithm. For lager $k$, Hassin and Rubinstein~\cite{hassin1997approximation} proposed a $3/4$-approximation algorithm for $k=4$, and a $\beta\cdot(1-1/k)$-approximation algorithm for $k\geq5$. On $\{0,1\}$-weighted graphs, Hassin and Schneider~\cite{hassin2013local} gave a $0.55$-approximation algorithm for $k=3$, the ratio for $k=3$ was improved to $3/4$~\cite{DBLP:journals/dam/Bar-NoyPRV18}, and Berman and Karpinski~\cite{DBLP:conf/soda/BermanK06} gave a $6/7$-approximation algorithm for $k=n$.

\textbf{Metric $k$PP.}
Li and Yu~\cite{li2023cyclepack} proposed a $3/4$-approximation algorithm for $k=3$, a $3/4$-approximation algorithm for $k=5$, and an $\alpha\cdot(1-1/k)$-approximation algorithm for $k\geq6$. The best-known result for $k=4$ is still $3/4$ due to the general $4$PP, by Hassin and Rubinstein~\cite{hassin1997approximation}. On $\{1,2\}$-weighted graphs, there is a $9/10$-approximation algorithm for $k=4$~\cite{DBLP:journals/jda/MonnotT08}.

General/metric $k$CP and $k$PP can be seen as a special case of the weighted $k$-set packing problem, which admits an approximation ratio of $\frac{1}{k-1}-\varepsilon$~\cite{arkin1998local}, $\frac{2}{k+1}-\varepsilon$~\cite{DBLP:journals/njc/Berman00}, and $\frac{2}{k+1-1/31850496}-\varepsilon$~\cite{DBLP:conf/stacs/Neuwohner21}. Recently, these results have been further improved (see~\cite{DBLP:conf/ipco/Neuwohner22,DBLP:conf/soda/Neuwohner23,thiery2023improved}). They can be used to obtain a $1/1.786\approx0.559$-approximation ratio for general 3CP~\cite{thiery2023improved}.

We will prove that for $k$CP/$k$PP a $\rho$-approximation algorithm on $\{0,1\}$-weighted graphs can be used directly to obtain an $(1+\rho)/2$-approximation algorithm on $\{1,2\}$-weighted graphs. Based on this, for a better view, we separately list the best-known results of $k$PP and $k$CP with $k\in\{3,4,n\}$ on $\{0,1\}$-weighted or $\{1,2\}$-weighted graphs in Table~\ref{special}. For $k\notin\{3,4,n\}$, the best-known results
%of $k$PP and $k$CP on $\{0,1\}$-weighted or $\{1,2\}$-weighted graphs 
are the corresponding results on general or metric graphs.

\begin{table}[ht]
\centering
\begin{tabular}{ccc}
\hline
   & $\{0,1\}$-weighted graphs & $\{1,2\}$-weighted graphs\\
\hline
  MAX $n$CP & $6/7$~\cite{DBLP:conf/soda/BermanK06} & $13/14$~\cite{DBLP:conf/soda/BermanK06}\\
  MAX $n$PP & $6/7$~\cite{DBLP:conf/soda/BermanK06} & $13/14$~\cite{DBLP:conf/soda/BermanK06}\\
\hline
  MAX $4$CP & $2/3$~\cite{li2023cyclepack} & $5/6$~\cite{li2023cyclepack}\\
  MAX $4$PP & $3/4$~\cite{hassin1997approximation} & $9/10$~\cite{DBLP:journals/jda/MonnotT08}\\
\hline
  MAX $3$CP & $3/5$~\cite{DBLP:journals/dam/Bar-NoyPRV18} & $4/5$~\cite{DBLP:journals/dam/Bar-NoyPRV18}\\
  MAX $3$PP & $3/4$~\cite{DBLP:journals/dam/Bar-NoyPRV18} & $7/8$~\cite{DBLP:journals/dam/Bar-NoyPRV18}\\
\hline
\end{tabular}
\caption{An overview of best-known approximation ratios for $k$PP and $k$CP with $k\in\{3,4,n\}$ on $\{0,1\}$-weighted or $\{1,2\}$-weighted graphs}
\label{special}
\end{table}

\subsection{Our Results}
In this paper, we study approximation algorithms for metric/general $k$CP and $k$PP. We mainly consider $k$ as a constant. The contributions can be summarized as follows.

Firstly, we consider metric $k$CP. We propose a $(7/8-0.125/k)(1-1/k)$-approximation algorithm for constant odd $k$ and a $7/8\cdot (1-1/k+\frac{1}{k(k-1)})$-approximation algorithm for even $k$, which improve the best-known approximation ratio of $3/5$ for $k=5$~\cite{li2023cyclepack} and $7/8\cdot(1-1/k)^2$ for $k\geq6$~\cite{li2023cyclepack}. Moreover, we propose an algorithm based on the maximum weight matching, which can further improve the approximation ratio from $17/25$ to $7/10$ for $k=5$. An illustration of the improved results for metric $k$CP with $k\geq 5$ can be seen in Table~\ref{res1}.

\begin{table}[ht]
\centering
\begin{tabular}{ccccc}
\hline
  Metric $k$CP & 5 & 6 & 7 & 8\\
\hline
  Previous Ratio~\cite{li2023cyclepack} & 0.600 & 0.607 & 0.642 & 0.669\\
\hline
  Our Ratio & $\boldsymbol{0.700}$ & $\boldsymbol{0.758}$ & $\boldsymbol{0.734}$ & $\boldsymbol{0.781}$\\
\hline
\end{tabular}
\caption{Improved approximation ratios for metric $k$CP with $k\geq 5$}
\label{res1}
\end{table}

Secondly, we consider metric $k$PP. We propose a $\frac{27k^2-48k+16}{32k^2-36k-24}$-approximation algorithm for even $10\geq k\geq 6$, which improves the best-known approximation ratio of $7/8\cdot(1-1/k)$~\cite{hassin1997approximation}. An illustration of the improved results for metric $k$PP with even $10\geq k\geq 6$ can be seen in Table~\ref{res2}.

\begin{table}[ht]
\centering
\begin{tabular}{cccc}
\hline
  Metric $k$PP & 6 & 8 & 10\\
\hline
  Previous Ratio~\cite{hassin1997approximation} & 0.729 & 0.765 & 0.787\\
\hline
  Our Ratio & $\boldsymbol{0.767}$ & $\boldsymbol{0.783}$ & $\boldsymbol{0.794}$\\
\hline
\end{tabular}
\caption{Improved approximation ratios for metric $k$PP with even $10\geq k\geq 6$}
\label{res2}
\end{table}

Thirdly, we focus on the case of $k=4$ for metric/general $k$CP and $k$PP. For metric 4CP, we propose a $5/6$-approximation algorithm, improving the best-known ratio $3/4$~\cite{li2023cyclepack}, and as a corollary, we also give a $7/8$-approximation algorithm on $\{1,2\}$-weighted graphs. For general 4CP, we propose a $3/4$-approximation algorithm, improving the best-known ratio $2/3$~\cite{li2023cyclepack}. For metric 4PP, we propose a $14/17$-approximation algorithm, improving the best-known ratio $3/4$~\cite{hassin1997approximation}. An illustration of the improved results for the case of $k=4$ can be seen in Table~\ref{res3}.

\begin{table}[ht]
\centering
\begin{tabular}{ccc}
\hline
& Metric Graphs & General Graphs\\
\hline
  4CP & $3/4~\cite{li2023cyclepack}\rightarrow \boldsymbol{5/6}$ & $2/3~\cite{li2023cyclepack}\rightarrow \boldsymbol{3/4}$\\
\hline
  4PP & $3/4~\cite{hassin1997approximation}\rightarrow \boldsymbol{14/17}$ & $3/4$~\cite{hassin1997approximation}\\
\hline
\end{tabular}
\caption{Improved results for the case of $k=4$}
\label{res3}
\end{table}

At last, we prove that for $k$CP/$k$PP a $\rho$-approximation algorithm on $\{0,1\}$-weighted graphs can be used to obtain a $(1+\rho)/2$-approximation algorithm on $\{1,2\}$-weighted graphs. Moreover, based on the $3/5$-approximation algorithm for 3CP on $\{0,1\}$-weighted graphs~\cite{DBLP:journals/dam/Bar-NoyPRV18}, we obtain a $9/11$-approximation algorithm for 3CP on $\{1,2\}$-weighted graphs. 

\subsection{Paper Organization}
The remaining parts of the paper are organized as follows. 
In Section~\ref{sec2}, we introduce some basic notations.
In Section \ref{sec3}, we consider metric $k$CP. In Section \ref{sec3.1}, we present a better black-box reduction from metric $k$CP to metric TSP, which has already led to an improved approximation ratio for $k\geq5$. In Section \ref{sec3.2}, by using some properties of the current-best approximation algorithm for metric TSP, we obtain a further improved approximation ratio. In Section \ref{sec3.2}, we consider a simple algorithm based on matching and show that the approximation ratio is better for $k=5$. 
In Section \ref{sec4}, we consider metric $k$PP and propose an improved approximation algorithm for even $10\geq k\geq 6$. Note that metric $k$PP is harder to improve, unlike metric $k$CP. 
In Section \ref{sec5}, we propose non-trivial approximation algorithms for metric/general $k$CP and $k$PP with the case of $k=4$. In Section \ref{sec5.1}, we obtain a better approximation algorithm for general 4CP. In Section \ref{sec5.2}, we obtain a better approximation algorithm for metric 4CP. In Section \ref{sec5.3}, we obtain a better approximation algorithm for metric 4PP. 
In Section \ref{sec7}, we reduce $k$CP/$k$PP on $\{1,2\}$-weighted graphs to $k$CP/$k$PP on $\{0,1\}$-weighted graphs, and give a $9/11$-approximation algorithm for 3CP on $\{1,2\}$-weighted graphs.
Finally, we make the concluding remarks in Section \ref{sec6}.

A preliminary version of this paper was 
%Partial results of this paper were 
presented at the 18th International Conference and Workshops on Algorithms and Computation (WALCOM 2024)~\cite{DBLP:conf/walcom/0001024}.

\section{Preliminaries}\label{sec2}
We use $G=(V, E)$ to denote an undirected complete graph with $n$ vertices such that $n\bmod k=0$. There is a non-negative weight function $w: E\to \mathbb{R}_{\geq0}$ on the edges in $E$. For an edge $uv\in E$, we use $w(u,v)$ to denote its weight. A graph is called a \emph{metric} graph if the weight function satisfies the triangle inequality; Otherwise, it is called a \emph{general} graph. For any weight function $w:X\to \mathbb{R}_{\geq0}$, we define $w(Y)=\sum_{x\in Y}w(x)$ for any $Y\subseteq X$.

Two subgraphs or subsets of edges of a graph are \emph{vertex-disjoint} if they do not appear a common vertex. We only consider simple paths and simple cycles with more than two vertices. The \emph{length} of a path/cycle is the number of vertices it contains. A \emph{cycle packing} is a set of vertex-disjoint cycles such that the length of each cycle is at least three and all vertices in the graph are covered. Given a cycle packing $\C$, we use $l(\C)$ to denote the minimum length of cycles in $\C$. We also use $\C^*$ to denote the maximum weight cycle packing.
A path (resp., cycle) on $k$ different vertices $\{v_1,v_2,\dots,v_k\}$ is called a \emph{$k$-path} (resp., \emph{$k$-cycle}), denoted by $v_1v_2\cdots v_k$ (resp., $v_1v_2\cdots v_kv_1$). A \emph{$k$-path packing} (resp., \emph{$k$-cycle packing}) in graph $G$ is a set of vertex-disjoint $n/k$ $k$-paths (resp., $k$-cycles) such that all vertices in the graph are covered. 
Note that we can obtain a $k$-cycle packing by completing every $k$-path of a $k$-path packing. Let $\P^*_k$ (resp., $\C^*_k$) denote the maximum weight $k$-path packing (resp., $k$-cycle packing). We can get $w(\C^*)\geq w(\C^*_k)$ for $k\geq 3$.

A $2$-path packing is usually called a \emph{matching} of size $n/2$. The maximum weight matching of size $n/2$ is denoted by $\M^*$.
An $n$-cycle is also called a \emph{Hamiltonian} cycle. MAX TSP is to find a maximum weight Hamiltonian cycle. Since we consider maximization problems, we simply use general/metric TSP to denote MAX TSP in general/metric graphs. We use $H^*$ to denote the maximum weight Hamiltonian cycle. For a $k$-path $P=v_{1}v_{2}\cdots v_{k}$ where $k$ is even, we define 
\[
\widetilde{w}(P)=\sum_{i=1}^{k/2}w(v_{2j-1},v_{2j}).
\]

\section{Approximation Algorithms for Metric $k$CP}\label{sec3}
In this section, we improve the approximation ratio for metric $k$CP with $k\geq 5$. We will first present a better black-box reduction from metric $k$CP to metric TSP, which is sufficient to improve the previous ratio for $k\geq 5$. Then, based on the approximation algorithm for metric TSP, we prove an improved approximation ratio. Finally, we consider a matching-based algorithm that can further improve the ratio of metric $5$CP.

\subsection{A Better Black-Box}\label{sec3.1}
Given an $\alpha$-approximation algorithm for metric TSP, Li and Yu~\cite{li2023cyclepack} proposed an $\alpha\cdot (1-1/k)^2$-approximation algorithm for metric $k$CP. We will show that the ratio can be improved to $\alpha\cdot (1-0.5/k)(1-1/k)$. Moreover, for even $k$, the ratio can be further improved to $\alpha\cdot (1-0.5/k)(1-1/k+\frac{1}{k(k-1)})$. We first consider a simple algorithm, denoted by Alg.1, which mainly contains three steps.

\medskip
\noindent\textbf{Step~1.} Obtain a Hamiltonian cycle $H$ using an $\alpha$-approximation algorithm for metric TSP;

\noindent\textbf{Step~2.} Get a $k$-path packing $\P_k$ with $w(\P_k)\geq (1-1/k)w(H)$ from $H$: we can obtain a $k$-path packing by deleting one edge per $k$ edges from $H$; since there are $(1-1/k)n$ edges in $\P_k$ and $n$ edges in $H$, if we carefully choose the initial edge, we can make sure that the weight of $\P_k$ is at least $(1-1/k)n\cdot(1/n)\cdot w(H)$, i.e., on average each edge has a weight of at least $(1/n)\cdot w(H)$.

\noindent\textbf{Step~3.} Obtain a $k$-cycle packing $\C_k$ by completing the $k$-path packing $\P_k$.
\medskip

To analyze the approximation quality, we will use the path patching technique, which has been used in some papers~\cite{DBLP:journals/ipl/HassinR02,kostochka1985polynomial,DBLP:journals/tcs/KowalikM09}. %By the triangle inequality, we have the following lemma.
\begin{lemma}[\cite{DBLP:journals/ipl/HassinR02,kostochka1985polynomial}]\label{patch}
Let $G$ be a metric graph. Given a cycle packing $\C$, there is a polynomial-time algorithm to generate a Hamiltonian cycle $H$ such that $w(H)\geq \lra{1-0.5/l(\C)}w(\C)$.
\end{lemma}

Since the length of every $k$-cycle in the maximum weight $k$-cycle packing $\C^*_k$ equals to $k$, we have $l(\C^*_k)=k$. By Lemma~\ref{patch}, we have the following lemma.
\begin{lemma}\label{lb1}
$w(H^*)\geq (1-0.5/k)w(\C^*_k)$.
\end{lemma}

\begin{theorem}\label{tm1}
Given an $\alpha$-approximation algorithm for metric TSP, for metric $k$CP, Alg.1 is a polynomial-time $\alpha\cdot (1-0.5/k)(1-1/k)$-approximation algorithm.
\end{theorem}
\begin{proof}
By the algorithm, we can easily get 
\[
w(\C_k)\geq w(\P_k)\geq (1-1/k)w(H)\geq\alpha\cdot(1-1/k)w(H^*).
\]
By Lemma~\ref{lb1}, we have $w(\C_k)\geq \alpha\cdot (1-0.5/k)(1-1/k)w(\C^*_k)$. Therefore, the algorithm achieves an approximation ratio of $\alpha\cdot (1-0.5/k)(1-1/k)$ for metric $k$CP.
\end{proof}

Next, we propose an improved $\alpha\cdot (1-0.5/k)(1-1/k+\frac{1}{k(k-1)})$-approximation algorithm for even $k$, denoted by Alg.2. The previous two steps of Alg.2 are the same as Alg.1. However, Alg.2 will obtain a better $k$-cycle packing in Step 3:

\medskip
\noindent\textbf{New Step~3.} For each $k$-path $P_i=v_{i1}v_{i2}\cdots v_{ik}\in\P_k$, we obtain $k-1$ $k$-cycles $\{C_{i1},\dots,C_{i(k-1)}\}$ where $C_{ij}=v_{i1}v_{i2}\cdots v_{ij}v_{ik}v_{i(k-1)}\cdots v_{i(j+1)}v_{i1}$ (See Figure~\ref{fig01} for an illustration); let $C_{ij_i}$ denote the maximum weight cycle from these cycles; return a $k$-cycle packing $\C_k=\{C_{ij_i}\}_{i=1}^{n/k}$.
\medskip
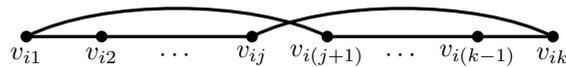
\begin{figure}[ht]
\centering
\begin{tikzpicture}
\filldraw [black]
(1,0) circle [radius=2pt]
(2,0) circle [radius=2pt]
(4,0) circle [radius=2pt]

(5,0) circle [radius=2pt]
(7,0) circle [radius=2pt]
(8,0) circle [radius=2pt];
\node (down) at (1,-0.25) {\small $v_{i1}$};
\node (down) at (2,-0.25) {\small $v_{i2}$};
\node (down) at (3,-0.25) {\small $\cdots$};
\node (down) at (4,-0.25) {\small $v_{ij}$};

\node (down) at (5,-0.25) {\small $v_{i(j+1)}$};
\node (down) at (6,-0.25) {\small $\cdots$};
\node (down) at (7,-0.25) {\small $v_{i(k-1)}$};
\node (down) at (8,-0.25) {\small $v_{ik}$};

\draw[very thick] (1,0) to (4,0);
\draw[very thick] (1,0) ..controls (2,0.5) and (4,0.5).. (5,0);
\draw[very thick] (5,0) to (8,0);
\draw[very thick] (4,0) ..controls (5,0.5) and (7,0.5).. (8,0);
\end{tikzpicture}
\caption{An illustration of the $k$-cycle $C_{ij}$ obtained from $P_i$, where $j\in\{1,2,\dots,k-1\}$}
\label{fig01}
\end{figure}

\begin{lemma}\label{path-cycle}
It holds that $w(\C_k)\geq\frac{k-2}{k-1}w(\P_k)+\frac{2}{k-1}\widetilde{w}(\P_k)$.
\end{lemma}
\begin{proof}
Since $C_{ij_i}$ is the maximum weight cycle from the $k-1$ cycles, we have
\begin{align*}
w(C_{ij_i})&\geq \frac{1}{k-1}\sum_{j=1}^{k-1}w(C_{ij})\\
&= \frac{1}{k-1}\sum_{j=1}^{k-1}(w(P_i)+w(v_{i1},v_{i(j+1)})+w(v_{ij},v_{ik})-w(v_{ij},v_{i(j+1)}))\\
&= \frac{1}{k-1}\lrA{(k-1)w(P_i)+\sum_{j=1}^{k-1}(w(v_{i1},v_{i(j+1)})+w(v_{ij},v_{ik}))-w(P_i)}\\
&= \frac{1}{k-1}\lrA{(k-2)w(P_i)+\sum_{j=1}^{k-1}(w(v_{i1},v_{i(j+1)})+w(v_{ij},v_{ik}))}.\\
\end{align*}
By the triangle inequality, we can get that
\begin{align*}
\sum_{j=1}^{k-1}w(v_{i1},v_{i(j+1)})=&\ w(v_{i1},v_{i2})+\sum_{j=2}^{k/2}(w(v_{i1},v_{i(2j-1)})+w(v_{i1},v_{i(2j)}))\\
\geq&\ w(v_{i1},v_{i2})+\sum_{j=2}^{k/2}w(v_{i(2j-1)},v_{i(2j)})\\
=&\ \sum_{j=1}^{k/2}w(v_{i(2j-1)},v_{i(2j)})
=\widetilde{w}(P_i).
\end{align*} 
Similarly, we can get $\sum_{j=1}^{k-1}w(v_{ij},v_{ik})\geq\widetilde{w}(P_i)$. Hence, we can get 
\begin{align*}
w(C_{ij_i})&\geq\frac{1}{k-1}\lrA{(k-2)w(P_i)+\sum_{j=1}^{k-1}(w(v_{i1},v_{i(j+1)})+w(v_{ij},v_{ik}))}\\
&\geq\frac{(k-2)w(P_i)+2\widetilde{w}(P_i)}{k-1}.    
\end{align*}
By doing this for all $k$-paths in $\P_k$, we can get a $k$-cycle packing $\C_k$ such that $w(\C_k)\geq \frac{(k-2)w(\P_k)+2\widetilde{w}(\P_k)}{k-1}$.
\end{proof}

\begin{theorem}\label{tm2}
Given an $\alpha$-approximation algorithm for metric TSP, for metric $k$CP with even $k$, Alg.2 is a polynomial-time $\alpha\cdot (1-0.5/k)(1-1/k+\frac{1}{k(k-1)})$-approximation algorithm.
\end{theorem}
\begin{proof}
Recall that all $k$-paths in $\P_k$ are obtained from the $\alpha$-approximate Hamiltonian cycle $H$. By deleting one edge per $k$ edges from a Hamiltonian cycle $H$ and choosing the initial edge carefully, we can get a $k$-path packing $\P_k$ such that 
\[
(k-2)w(\P_k)+2\widetilde{w}(\P_k)\geq\frac{(k-2)(k-1)+k}{k}w(H)=\frac{(k-1)^2+1}{k}w(H)
\]
since $(k-2)w(\P_k)+2\widetilde{w}(\P_k)$ contains the weight of $\frac{n(k-2)(k-1)+nk}{k}$ (multi-)edges in $H$. By Lemma~\ref{path-cycle}, we can obtain a $k$-cycle packing $\C_k$ such that
\begin{align*}
w(\C_k)&\geq\frac{(k-2)w(\P_k)+2\widetilde{w}(\P_k)}{k-1}\\
&\geq\frac{(k-1)^2+1}{k(k-1)}w(H)=\lrA{1-1/k+\frac{1}{k(k-1)}}w(H). 
\end{align*}
Since $w(H)\geq \alpha\cdot w(H^*)\geq \alpha \cdot (1-0.5/k)w(\C^*_k)$ by Lemma~\ref{lb1}, we have $w(\C_k)\geq \alpha \cdot (1-0.5/k)(1-1/k+\frac{1}{k(k-1)})w(\C^*_k)$, and the approximation ratio is $\alpha \cdot (1-0.5/k)(1-1/k+\frac{1}{k(k-1)})$.
\end{proof}

Note that for metric TSP there is a randomized $(7/8-O(1/\sqrt{n}))$-approximation algorithm~\cite{DBLP:journals/ipl/HassinR02}, a deterministic $(7/8-O(1/\sqrt[3]{n}))$-approximation algorithm~\cite{DBLP:journals/jco/ChenN07}, and a deterministic $7/8$-approximation algorithm~\cite{DBLP:journals/tcs/KowalikM09}. By Theorem~\ref{tm2}, we obtain an approximation ratio of $7/8\cdot (1-0.5/k)(1-1/k)$ for metric $k$CP with odd $k$, and $7/8\cdot (1-0.5/k)(1-1/k+\frac{1}{k(k-1)})$ for metric $k$CP with even $k$.

\subsection{A Further Improvement}\label{sec3.2}
In this subsection, we show that the approximation ratio of Alg.2 can be further improved based on the properties of the current best $7/8$-approximation algorithm for metric TSP~\cite{DBLP:journals/tcs/KowalikM09}. We recall the following result.

\begin{lemma}[\cite{DBLP:journals/tcs/KowalikM09}]\label{7/8}
Let $G$ be a metric graph with even $n$. There is a polynomial-time algorithm to generate a Hamiltonian cycle $H$ in $G$ such that
\[
w(H)\geq \frac{5}{8}w(\C^*)+\frac{1}{2}w(\M^*).
\]
\end{lemma}

For any $k$-cycle packing with $k$ being even or Hamiltonian cycle with an even number of vertices, the edges can be decomposed into two edge-disjoint matchings of size $n/2$. We can get the following bounds.

\begin{lemma}\label{lb2}
It holds that $w(\M^*)\geq \frac{1}{2}w(\C^*_k)$ for even $k$ and $w(\M^*)\geq \frac{1}{2}w(H^*)$ for even $n$.
\end{lemma}

Note that for metric $k$CP with even $k$, the number of vertices is always even since it satisfies $n\bmod k=0$. But for odd $k$, the number may be odd, and then there may not exist a matching of size $n/2$. Since we mainly consider the improvements for constant $k$, for the case of odd $k$ and $n$, we can first use $n^{O(k)}=n^{O(1)}$ time to enumerate a $k$-cycle in $\C^*_k$, and then consider an approximate $k$-cycle packing in the rest graph. The approximation ratio preserves. Hence, we assume that $n$ is even for the case of constant $k$. 

\begin{theorem}\label{tm3}
For metric $k$CP, there is a polynomial-time $(7/8-0.125/k)(1-1/k)$-approximation algorithm for constant odd $k$ and a polynomial-time $7/8\cdot (1-1/k+\frac{1}{k(k-1)})$-approximation algorithm for even $k$.
\end{theorem}
\begin{proof}
Consider constant odd $k$. Since we assume that the number of vertices is even, by Lemmas~\ref{7/8} and \ref{lb2}, we can get a Hamiltonian cycle $H$ such that
\[
w(H)\geq(5/8)w(\C^*)+(1/2)w(\M^*)\geq(5/8)w(\C^*)+(1/4)w(H^*).
\]
Recall that $w(H^*)\geq (1-0.5/k)w(\C^*_k)$ by Lemma~\ref{lb1}, and $w(\C^*)\geq w(\C^*_k)$. We can get that 
\[
w(H)\geq(5/8)w(\C^*_k)+(1/4)(1-0.5/k)w(\C^*_k)=(7/8-0.125/k)w(\C^*_k).
\]
By the proof of Theorem~\ref{tm1}, we can get a $k$-cycle packing $\C_k$ with $w(\C_k)\geq (1-1/k)w(H)$. So, we have $w(\C_k)\geq (7/8-0.125/k)(1-1/k)w(\C^*_k)$, and the approximation ratio is $(7/8-0.125/k)(1-1/k)$ for constant odd $k$.

Consider even $k$. Recall that $w(\C^*)\geq w(\C^*_k)$. Since the number of vertices is even, by Lemmas~\ref{7/8} and \ref{lb2}, we can get a Hamiltonian cycle $H$ such that
\[
w(H)\geq(5/8)w(\C^*)+(1/2)w(\M^*)\geq(5/8)w(\C^*_k)+(1/4)w(\C^*_k)=(7/8)w(\C^*_k).
\]
Similarly, by the proof of Theorem~\ref{tm2}, we can get a $k$-cycle packing $\C_k$ such that $w(\C_k)\geq (1-1/k+\frac{1}{k(k-1)})w(H)$. Therefore, we have $w(\C_k)\geq 7/8\cdot (1-1/k+\frac{1}{k(k-1)})w(\C^*_k)$, and the approximation ratio is $7/8\cdot (1-1/k+\frac{1}{k(k-1)})$ for even $k$.
\end{proof}

\subsection{An Improved Algorithm Based on Matching}\label{sec3.3}
Consider metric $k$CP with odd $k$. By deleting the least weighted edge from every $k$-cycle in $\C^*_k$, we can get a $k$-path packing $\P_k$ with $w(\P_k)\geq (1-1/k)w(\C^*_k)$. Note that $\P_k$ can be decomposed into two edge-disjoint matchings of size $p\coloneqq(n/k)\cdot (k-1)/2$. Let $\M^*_p$ be the maximum weight matching of size $p$, which can be computed in polynomial time~\cite{gabow1974implementation,lawler1976combinatorial}. Then, we can get 
\[
2w(\M^*_p)\geq w(\P_k)\geq (1-1/k)w(\C^*_k).
\]
Note that there are also $n/k$ \emph{isolated vertices} not covered by $\M^*_p$. Next, we construct a $k$-cycle packing using $\M^*_p$ with the isolated vertices. The algorithm, denoted by Alg.3, is shown as follows.

\medskip
\noindent\textbf{Step~1.} Arbitrarily partition the $p$ edges of $\M^*_p$ into $n/k$ sets with the same size, denoted by $\S_1,\S_2,\dots,\S_{n/k}$. Note that each edge set contains $m\coloneqq(k-1)/2$ edges. For each of the $n/k$ edge sets, arbitrarily assign an isolated vertex.

\noindent\textbf{Step~2.} Consider an arbitrary edge set $\S_i=\{e_1,e_2,\dots,e_{m}\}$ with the isolated vertex $v$. Assume w.o.l.g. that $w(e_1)\geq w(e_m)\geq w(e_i)$ for $2\leq i<m$, i.e., $w(e_1)+w(e_m)\geq (2/m)w(\S_i)$. Orient each edge $e_i$ uniformly at random from the two choices. Let $t_i$ (resp., $h_i$) denote the tail (resp., head) vertex of $e_i$. Construct a $k$-cycle $C_i$ such that $C_i=vt_1h_1t_2h_2\cdots t_mh_mv$.

\noindent\textbf{Step~3.} Get a $k$-cycle packing $\C_k$ by packing the $k$-cycles from the edge sets and the isolated vertices.
\medskip

Alg.3 can be derandomized efficiently by the method of conditional expectations~\cite{williamson2011design}.

Next, we analyze the expected weight of $C_i=vt_1h_1t_2h_2\cdots t_mh_mv$, obtained from the edge set $\S_i$ and the isolated vertex $v$.
\begin{lemma}\label{lb3}
It holds that $\EE{w(v,t_1)}\geq\frac{1}{2}w(e_1)$, $\EE{w(v,h_m)}\geq\frac{1}{2}w(e_m)$, and $\EE{w(h_i,t_{i+1})}\geq\frac{1}{4}(w(e_i)+w(e_{i+1}))$ for $1\leq i<m$.
\end{lemma}
\begin{proof}
Consider $\EE{w(v,t_1)}$. Since we orient the edge $e_1$ uniformly at random, each vertex of $e_1$ has a probability of $1/2$ being $t_1$. Hence, we can get $\EE{w(v,t_1)}=\frac{1}{2}\sum_{u\in e_1}w(v,u)\geq \frac{1}{2}w(e_1)$ by the triangle inequality. Similarly, we can get $\EE{w(v,h_m)}\geq\frac{1}{2}w(e_m)$.

Consider $\EE{w(h_i,t_{i+1})}$. We can get $\EE{w(h_i,t_{i+1})}=\frac{1}{4}\sum_{u\in e_i}\sum_{w\in e_{i+1}}w(u,w)$. Let $e_i=u'u''$ and $e_{i+1}=o'o''$. By the triangle inequality, we can get that
\begin{align*}
\sum_{u\in e_i}\sum_{w\in e_{i+1}}w(u,w)=&\ w(u',o')+w(u',o'')+w(u'',o')+w(u'',o'')\\
\geq&\ w(o',o'')+w(u',u'')
=\ w(e_i)+w(e_{i+1}).
\end{align*} 
Therefore, $\EE{w(h_i,t_{i+1})}\geq\frac{1}{4}(w(e_i)+w(e_{i+1}))$ for $1\leq i<m$.
\end{proof}

\begin{lemma}\label{lb4}
It holds that $\EE{w(C_i)}\geq\frac{3m+1}{2m}w(\S_i)$.
\end{lemma}
\begin{proof}
Note that
\begin{align*}
w(C_i)=&\ w(v,t_1)+w(v,h_m)+\sum_{i=1}^{m-1}(w(t_i,h_i)+w(h_i,t_{i+1}))\\
=&\ w(\S_i)+w(v,t_1)+w(v,h_m)+\sum_{i=1}^{m-1}w(h_i,t_{i+1}).
\end{align*}
We can get that
\begin{align*}
\EE{w(C_i)}\geq&\ w(\S_i)+\frac{1}{2}(w(e_1)+w(e_m))+\frac{1}{4}\sum_{i=1}^{m-1}(w(e_i)+w(e_{i+1}))\\
=&\ w(\S_i)+\frac{1}{2}(w(e_1)+w(e_m))+\frac{1}{2}w(\S_i)-\frac{1}{4}(w(e_1)+w(e_m))\\
=&\ \frac{3}{2}w(\S_i)+\frac{1}{4}(w(e_1)+w(e_m))\\
\geq&\ \lrA{\frac{3}{2}+\frac{1}{2m}}w(\S_i)
=\frac{3m+1}{2m}w(\S_i),
\end{align*}
where the first inequality follows from Lemma~\ref{lb3}, and the second from $w(e_1)+w(e_m)\geq (2/m)w(\S_i)$ by the algorithm.
\end{proof}

\begin{theorem}\label{tm4}
For metric $k$CP with odd $k$, Alg.3 is a polynomial-time $(3/4-0.25/k)$-approximation algorithm.
\end{theorem}
\begin{proof}
Recall that $2w(\M^*_p)\geq (1-1/k)w(\C^*_k)$ and $\M^*_p=\bigcup_{i=1}^{n/k}\S_i$. Using a derandomization based on conditional expectations~\cite{williamson2011design}, by Lemma~\ref{lb4}, we can get that 
\[
w(\C_k)\geq \sum_{i=1}^{n/k}\frac{3m+1}{2m}w(\S_i)=\frac{3m+1}{2m}w(\M^*_p)\geq \frac{3m+1}{4m}\lrA{1-\frac{1}{k}}w(\C^*_k).
\]
Since $m=(k-1)/2$, we can get an approximation ratio of $\frac{3m+1}{4m}(1-\frac{1}{k})=3/4-0.25/k$.
\end{proof}

By Theorem~\ref{tm4}, we obtain a $7/10$-approximation algorithm for metric $5$CP, which improves the previous ratio $17/25$ in Theorem~\ref{tm3}, and the ratio $3/5$ in~\cite{li2023cyclepack}. Note that for metric $k$CP with odd $k>5$, the result in Theorem~\ref{tm4} is worse than Theorem~\ref{tm3}.

\begin{corollary}
For metric $5$CP, Alg.3 is a polynomial-time $7/10$-approximation algorithm.
\end{corollary}

The analysis of our algorithm is tight. Figure~\ref{fig02} shows an example, where there are $n=25$ vertices, the weight of each solid edge is 2, and the weight of each omitted edge is 1. Note that the maximum weight 5-cycle packing $\{u_{(i,1)}u_{(i,2)}u_{(i,3)}u_{(i,4)}u_{(i,5)}u_{(i,1)}\}_{i=1}^{5}$ has a weight of 100. The algorithm may find a maximum weight matching $\M^*_p=\{u_{(i,1)}u_{(i,2)},u_{(i,3)}u_{(i,4)}\}_{i=1}^{5}$ of size $p=(n/k)\cdot (k-1)/2=10$, and get a set of two edges and an isolated vertex $\{u_{(i,1)}u_{(i,2)},u_{(i\bmod 5+1,3)}u_{(i\bmod 5+1,4)}, u_{(i+1\bmod5+1,5)}\}$ for each $1\leq i\leq5$, which can only form a $5$-cycle with a weight of $7$. Then, the weight of obtained 5-cycle packing is $35$, and the approximation ratio is $35/50=7/10$.

\begin{figure}[ht]
\centering
\begin{tikzpicture}[scale=0.8]
\foreach \y in {1,...,5}
{
    \foreach \x in {1,...,5}
    {
        \filldraw [black] (\y,0.75*\x) circle [radius=2pt];
        \draw (\x,0.75*6-0.75*\y+0.3) node{$u_{(\y,\x)}$};
    }
    \draw[very thick] (1,0.75*\y) to (5,0.75*\y);
    \draw[very thick] (1,0.75*\y) ..controls (2,0.75*\y+0.2) and (4,0.75*\y+0.2).. (5,0.75*\y);
}
\end{tikzpicture}
\caption{A tight example of the $7/10$-approximation algorithm for metric 5CP}
\label{fig02}
\end{figure}
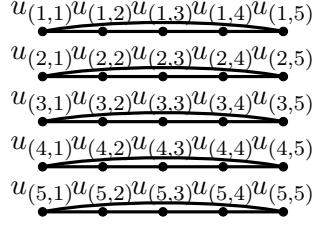

For metric $3$CP, the approximation ratio in Theorem~\ref{tm4} achieves $2/3\approx 0.66666$. However, there exist better results: a randomized $(0.66768-\varepsilon)$-approximation algorithm~\cite{DBLP:journals/jco/ChenCLWZ21} and a deterministic $(0.66835-\varepsilon)$-approximation algorithm~\cite{DBLP:journals/corr/abs-2402-08216}.
%For metric $3$CP, the approximation ratio in Theorem~\ref{tm4} achieves $2/3\approx 0.66666$. However, there exist better results: a randomized $(0.66768-\varepsilon)$-approximation algorithm~\cite{DBLP:journals/jco/ChenCLWZ21} and a deterministic $(0.66835-\varepsilon)$-approximation algorithm~\cite{2023ttp2}.
It is also worth noting that using a similar framework we can obtain a $3/4$-approximation algorithm for metric $k$CP with even $k$. The ratio is the same as the $3/4$-approximation algorithm for metric $4$CP~\cite{li2023cyclepack}, and worse than the result in Theorem~\ref{tm3} for metric $k$CP with even $k>4$. So, we omit it. In Section~\ref{sec5}, we will obtain an improved $5/6\approx0.833$-approximation algorithm for metric $4$CP. Next, we first consider metric $k$PP.

\section{Approximation Algorithms for Metric $k$PP}\label{sec4}
In this section, we consider metric $k$PP. Using a reduction from metric $k$PP to metric TSP, metric $k$PP admits a $7/8\cdot (1-1/k)$-approximation algorithm~\cite{hassin1997approximation}. Note that unlike metric $k$CP it is not easy to construct a better black box to improve the ratio. However, we will combine the properties of the $7/8$-approximation algorithm for metric TSP with an algorithm based on matching to obtain a better approximation ratio for even $6\leq k\leq 10$. So, in this section, we assume that $k$ is even.

The first algorithm, denoted by Alg.4, is to use the reduction from metric $k$PP to metric TSP~\cite{hassin1997approximation}.

\medskip
\noindent\textbf{Step~1.} Obtain a Hamiltonian cycle $H$ using the $7/8$-approximation algorithm for metric TSP;

\noindent\textbf{Step~2.} Get a $k$-path packing $\P_k$ with $w(\P_k)\geq (1-1/k)w(H)$ from $H$ using the same method in Step~2 of Alg.1.
\medskip

For every $P_i=v_{i1}v_{i2}\cdots v_{ik}\in\P^*_k$, let 
\[
\E_{i}^{'}=\{v_{i(2j-1)}v_{i(2j)}\}_{j=1}^{k/2}\quad \mbox{and}\quad\E_{i}^{''}=\{v_{i(2j)}v_{i(2j+1)}\}_{j=1}^{(k-2)/2}.
\]
Then, we can obtain a matching $\M_{n/2}=\bigcup_{i}\E_{i}^{'}$ of size $n/2$ and a matching $\M_p=\bigcup_{i}\E_{i}^{''}$ of size $p\coloneqq(n/k)\cdot (k-2)/2$. Note that $w(\M_{n/2})+w(\M_p)=w(\P^*_k)$. We have the following bounds.

\begin{lemma}\label{p1}
$w(\C^*_k)\geq\frac{k-2}{k-1}w(\P^*_k)+\frac{2}{k-1}w(\M_{n/2})$.
\end{lemma}
\begin{proof}
For every $P_i=v_{i1}v_{i2}\cdots v_{ik}\in\P^*_k$, it holds that $\widetilde{w}(P_i)=\sum_{j=1}^{k/2}w(v_{i(2j-1)},v_{i(2j)})=w(\E_{i}^{'})$. Hence, $\widetilde{w}(\P^*_k)=w(\M_{n/2})$. By the proof of Lemma~\ref{path-cycle}, we can get a $k$-cycle packing $\C_k$ from $\P^*_k$ such that 
\[
w(\C_k)\geq\frac{k-2}{k-1}w(\P^*_k)+\frac{2}{k-1}\widetilde{w}(\P^*_k)=\frac{k-2}{k-1}w(\P^*_k)+\frac{2}{k-1}w(\M_{n/2}).
\]
Since $\C^*_k$ is the maximum weight $k$-cycle packing, we have $w(\C^*_k)\geq w(\C_k)$.
\end{proof}

\begin{lemma}\label{p2}
$w(\P_k)\geq\frac{5k-10}{8k}w(\P^*_k)+\frac{2k+3}{4k}w(\M_{n/2})$.
\end{lemma}
\begin{proof}
By Lemma~\ref{7/8}, we can get $w(H)\geq \frac{5}{8}w(\C^*)+\frac{1}{2}w(\M^*)$. Since $w(\P_k)\geq \frac{k-1}{k}w(H)$, $w(\C^*)\geq w(\C^*_k)$, and $w(\M^*)\geq w(\M_{n/2})$, we have 
\begin{align*}
w(\P_k)\geq&\ \frac{k-1}{k}\lrA{\frac{5}{8}w(\C^*_k)+\frac{1}{2}w(\M_{n/2})}\\
\geq&\ \frac{k-1}{k}\lrA{\frac{5}{8}\lrA{\frac{k-2}{k-1}w(\P^*_k)+\frac{2}{k-1}w(\M_{n/2})}+\frac{1}{2}w(\M_{n/2})}\\
=&\ \frac{5k-10}{8k}w(\P^*_k)+\frac{2k+3}{4k}w(\M_{n/2}),
\end{align*}
where the second inequality follows from Lemma~\ref{p1}.
\end{proof}

Next, we propose an algorithm, denoted by Alg.5, to obtain another $k$-path packing $\P'_k$, which can be used to make a trade-off with $\P_k$. The framework of Alg.5 is similar to Alg.3 in Section~\ref{sec3.3}. Let $\M^*_p$ denote the maximum weight matching of size $p=(n/k)\cdot (k-2)/2$, which can be computed in polynomial time~\cite{gabow1974implementation,lawler1976combinatorial}. Note that $w(\M^*_p)\geq w(\M_p)$. There are also $2n/k$ \emph{isolated vertices} not covered by $\M^*_p$. Next, we construct a $k$-path packing using $\M^*_p$ with isolated vertices. Alg.5 mainly contains three steps.

\medskip
\noindent\textbf{Step~1.} Arbitrarily partition the $p$ edges of $\M^*_p$ into $n/k$ sets with the same size, denoted by $\S_1,\S_2,\dots,\S_{n/k}$. Note that each edge set contains $m\coloneqq(k-2)/2$ edges. For each of the $n/k$ edge sets, arbitrarily assign two isolated vertices.

\noindent\textbf{Step~2.} Consider an arbitrary edge set $\S_i=\{e_1,e_2,\dots,e_{m}\}$ with the two isolated vertices $u$ and $v$. Assume w.o.l.g. that $w(e_1)\geq w(e_m)\geq w(e_i)$ for $2\leq i<m$, i.e., $w(e_1)+w(e_m)\geq (2/m)w(\S_i)$. Orient each edge $e_i$ uniformly at random from the two choices. Let $t_i$ (resp., $h_i$) denote the tail (resp., head) vertex of $e_i$. Construct a $k$-path $P'_i$ such that $P'_i=ut_1h_1t_2h_2\cdots t_mh_mv$.

\noindent\textbf{Step~3.} Get a $k$-path packing $\P'_k$ by packing the $k$-paths from the edge sets and the isolated vertices.
\medskip

Next, we analyze the expected weight of $P'_i=ut_1h_1t_2h_2\cdots t_mh_mv$, obtained from the edge set $\S_i$ and the two isolated vertices $u$ and $v$.
\begin{lemma}\label{p3}
It holds that $\EE{w(u,t_1)}\geq\frac{1}{2}w(e_1)$, $\EE{w(v,h_m)}\geq\frac{1}{2}w(e_m)$, and $\EE{w(h_i,t_{i+1})}\geq\frac{1}{4}(w(e_i)+w(e_{i+1}))$ for $1\leq i<m$.
\end{lemma}
\begin{proof}
See the proof of Lemma~\ref{lb3}.
\end{proof}

\begin{lemma}\label{p4}
It holds that $\EE{w(\P'_k)}\geq\frac{3k-4}{2k-4}w(\M^*_p)$.
\end{lemma}
\begin{proof}
It is sufficient to prove $\EE{w(P'_i)}\geq\frac{3m+1}{2m}w(\S_i)$. Note that
\[
w(P'_i)=w(\S_i)+w(u,t_1)+w(v,h_m)+\sum_{i=1}^{m-1}w(h_i,t_{i+1}).
\]
Since $m=\frac{k-2}{2}$, we can get that
\begin{align*}
\EE{w(P'_i)}\geq&\ w(\S_i)+\frac{1}{2}(w(e_1)+w(e_m))+\frac{1}{4}\sum_{i=1}^{m-1}(w(e_i)+w(e_{i+1}))\\
=&\ w(\S_i)+\frac{1}{2}(w(e_1)+w(e_m))+\frac{1}{2}w(\S_i)-\frac{1}{4}(w(e_1)+w(e_m))\\
=&\ \frac{3}{2}w(\S_i)+\frac{1}{4}(w(e_1)+w(e_m))\\
\geq&\ \lrA{\frac{3}{2}+\frac{1}{2m}}w(\S_i)
=\frac{3k-4}{2k-4}w(\S_i),
\end{align*}
where the first inequality follows from Lemma~\ref{p3}, and the second from $w(e_1)+w(e_m)\geq (2/m)w(\S_i)$ by the algorithm.
%Note that $w(P'_i)=w(\S_i)+w(u,t_1)+w(v,h_m)+\sum_{i=1}^{m-1}w(h_i,t_{i+1})$. Recall that $w(e_1)+w(e_m)\geq (2/m)w(\S_i)$. By Lemma~\ref{p3}, $\EE{w(P'_i)}\geq w(\S_i)+\frac{1}{2}(w(e_1)+w(e_m))+\frac{1}{4}\sum_{i=1}^{m-1}(w(e_i)+w(e_{i+1}))=w(\S_i)+\frac{1}{2}(w(e_1)+w(e_m))+\frac{1}{2}w(\S_i)-\frac{1}{4}(w(e_1)+w(e_m))=\frac{3}{2}w(\S_i)+\frac{1}{4}(w(e_1)+w(e_m))\geq (\frac{3}{2}+\frac{1}{2m})w(\S_i)=\frac{3k-4}{2k-4}w(\S_i)$ by $m=\frac{k-2}{2}$.
\end{proof}

Alg.5 can also be derandomized efficiently by the method of conditional expectations~\cite{williamson2011design}. 

\begin{lemma}\label{p5}
$w(\P'_k)\geq\frac{3k-4}{2k-4}w(\M_p)$.
\end{lemma}
\begin{proof}
Recall that $w(\M^*_p)\geq w(\M_p)$. By Lemma~\ref{p4} with a derandomization, we can get $w(\P'_k)\geq\frac{3k-4}{2k-4}w(\M_p)$.
\end{proof}

\begin{theorem}\label{tmp}
For metric $k$PP with even $k$, there is a polynomial-time $\frac{27k^2-48k+16}{32k^2-36k-24}$-approximation algorithm.
\end{theorem}
\begin{proof}
The algorithm returns a better $k$-path packing from $\P_k$ and $\P'_k$. Recall that $w(\M_{n/2})+w(\M_p)=w(\P^*_k)$. By Lemmas~\ref{p2} and \ref{p5}, we have 
\begin{align*}
&\lrA{\frac{3k-4}{2k-4}+\frac{2k+3}{4k}}\cdot\max\{w(\P_k),w(\P'_k)\}\\
&\ \geq \frac{3k-4}{2k-4}w(\P_k)+\frac{2k+3}{4k}w(\P'_k)\\
&\ \geq\frac{3k-4}{2k-4}\cdot\frac{5k-10}{8k}w(\P^*_k)+\frac{3k-4}{2k-4}\cdot\frac{2k+3}{4k}(w(\M_{n/2})+w(\M_p))\\
&\ =\lrA{\frac{3k-4}{2k-4}\cdot\frac{5k-10}{8k}+\frac{3k-4}{2k-4}\cdot\frac{2k+3}{4k}}w(\P^*_k)\\
&\ =\frac{3k-4}{2k-4}\cdot\frac{9k-4}{8k}w(\P^*_k).
\end{align*}
Hence, $\max\{w(\P_k),w(\P'_k)\}\geq\frac{3k-4}{2k-4}\cdot\frac{9k-4}{8k}/\lra{\frac{3k-4}{2k-4}+\frac{2k+3}{4k}}w(\P^*_k)=\frac{27k^2-48k+16}{32k^2-36k-24}w(\P^*_k)$. %The ratio is $\frac{27k^2-48k+16}{32k^2-36k-24}$.
\end{proof}

The approximation ratio in Theorem~\ref{tmp} is better than $7/8\cdot (1-1/k)$ for even $10\geq k\geq 6$ (see Table~\ref{res2}). For metric $4$PP, the ratio is even worse than the ratio $3/4$ in~\cite{hassin1997approximation}. But, in the next section, we will show an improved $14/17\approx 0.823$-approximation algorithm for this case.

\section{Approximation Algorithms for the Case of $k=4$}\label{sec5}
In this section, we study the case of $k=4$ for metric/general $k$CP and $k$PP. For metric 4CP, we improve the best-known ratio from $3/4$~\cite{li2023cyclepack} to $5/6$. For general 4CP, we improve the best-known ratio from $2/3$~\cite{li2023cyclepack} to $3/4$. For metric 4PP, we improve the best-known ratio from $3/4$~\cite{hassin1997approximation} to $14/17$.

\subsection{General 4CP}\label{sec5.1}
Some structural properties between the minimum weight $4$-cycle packing and the minimum weight matching of size $n/2$ were proved in \cite{DBLP:journals/jda/MonnotT08,DBLP:journals/corr/abs-2210-16534}. The properties can be directly extended to the maximum weight $4$-cycle packing and the maximum weight matching of size $n/2$.

%Zhao and Xiao~\cite{DBLP:journals/corr/abs-2210-16534} observed some structural properties of the minimum weight $4$-cycle packing and the minimum weight matching of size $n/2$. In fact, these properties even hold for the maximum weight $4$-cycle packing $\C^*_4$ and the maximum weight matching $\M^*$ of size $n/2$. %Firstly, we show some good properties of the maximum weight $4$-cycle packing $\C^*_4$ and the maximum weight matching $\M^*$ of size $n/2$.

\begin{lemma}[\cite{DBLP:journals/jda/MonnotT08,DBLP:journals/corr/abs-2210-16534}]\label{mod4}
Given the maximum weight 4-cycle packing $\C^*_4$ and the maximum weight matching $\M^*$ of size $n/2$, there is a way to color edges in $\C^*_4$ with red and blue such that
\begin{enumerate}
    \item [(1)] the blue (resp., red) edges form a matching of size $n/2$ $\M_b$ (resp., $\M_r$);
    \item [(2)] $\C^*_4=\M_b\cup\M_r$;
    \item [(3)] $\M_b\cup\M^*$ is a cycle packing such that the length of every cycle is divisible by 4.
\end{enumerate}
\end{lemma}

Next, we describe the approximation algorithm for general 4CP, denoted by Alg.6.

\medskip
\noindent\textbf{Step~1.} Find a maximum weight matching $\M^*$ of size $n/2$.

\noindent\textbf{Step~2.} Construct a multi-graph $G/\M^*$ such that there are $n/2$ super-vertices one-to-one corresponding to the $n/2$ edges in $\M^*$, i.e., there is a function $f$, and for two super-vertices $f(e_i),f(e_j)$ such that $e_i,e_j\in\M^*$, there are four super-edges $f(e_i)f(e_j)$ between them, corresponding to the four edges $uv$ with a weight of $w(u,v)$ ($u\in e_i, v\in e_j$).

\noindent\textbf{Step~3.} Find a maximum weight matching $\M^{**}_{n/4}$ of size $n/4$ in graph $G/\M^*$. Note that $\M^*\cup\M^{**}_{n/4}$ corresponds to a $4$-path packing $\P_4$ in graph $G$.

\noindent\textbf{Step~4.} Obtain a $4$-cycle packing $\C_4$ by completing the $4$-path packing $\P_4$.
\medskip

We can get that $w(\C_4)\geq w(\P_4)=w(\M^*)+w(\M^{**}_{n/4})$.

\begin{lemma}\label{lb5}
$w(\M^{**}_{n/4})\geq \frac{1}{2}w(\M_b)$.
\end{lemma}
\begin{proof}
For every cycle of $\M_b\cup\M^*$, since the length is divisible by 4, the number of blue edges is even, and then it corresponds to an even cycle in $G/\M^*$. Hence, we can obtain two matchings of size $n/4$ by decomposing the blue edges in $G/\M^*$. Since $\M^{**}_{n/4}$ is the maximum weight matching, we can get $w(\M^{**}_{n/4})\geq \frac{1}{2}w(\M_b)$.
\end{proof}

\begin{lemma}\label{useful}
$w(\P_4)\geq \frac{1}{2}w(\M^*)+\frac{1}{2}w(\C^*_4)$.
\end{lemma}
\begin{proof}
Recall that $w(\P_4)=w(\M^*)+w(\M^{**}_{n/4})$. By Lemma~\ref{lb5}, we have 
\[
w(\P_4)\geq w(\M^*)+\frac{1}{2}w(\M_b).
\]
Note that $w(\M^*)\geq \max\{w(\M_b),w(\M_r)\}$ and $\C^*_4=\M_b\cup\M_r$. Hence, we can get 
\begin{align*}
w(\P_4)&\geq w(\M^*)+\frac{1}{2}w(\M_b)\geq \frac{1}{2}w(\M^*)+\frac{1}{2}w(\M_r)+\frac{1}{2}w(\M_b)\\
&=\frac{1}{2}w(\M^*)+\frac{1}{2}w(\C^*_4).    
\end{align*}
\end{proof}

\begin{theorem}\label{tm5}
For general 4CP, Alg.6 is a polynomial-time $3/4$-approximation algorithm.
\end{theorem}
\begin{proof}
Recall that $w(\C_4)\geq w(\P_4)$.
By Lemma~\ref{useful}, we have $w(\P_4)\geq\frac{1}{2}w(\M^*)+\frac{1}{2}w(\C^*_4)$. Moreover, by Lemma~\ref{lb2}, we have $w(\M^*)\geq\frac{1}{2}w(\C^*_4)$. So, $w(\C_4)\geq \frac{3}{4}w(\C^*_4)$.
\end{proof}

The analysis is tight, even on $\{0,1\}$-weighted graphs. Figure~\ref{fig04} shows an example, where there are $n=12$ vertices, the weight of each solid edge is 1, and the weight of each omitted edge is 0. Note that the maximum weight 4-cycle packing $\{v_2v_3v_8v_9v_2, v_4v_5v_{10}v_{11}v_4, v_6v_7v_{12}v_1v_6\}$ has a weight of 12. The algorithm may find a maximum weight matching $\M^*=\{v_1v_2,v_3v_4,\dots,v_{11}v_{12}\}$ of size $n/2$ with a weight of 6, and get a maximum weight 4-cycle packing $\C_4=\{v_1v_2v_3v_4v_1, v_5v_6v_7v_8v_5, v_9v_{10}v_{11}v_{12}v_9\}$ containing the edges of $\M^*$. Then, the weight of $\C_4$ is $9$, and the approximation ratio is $9/12=3/4$.

\begin{figure}[ht]
\centering
\begin{tikzpicture}[scale=0.8]
\filldraw[black]
(0.5, 1.866) circle [radius=2pt]
(1.366, 1.366) circle [radius=2pt]
(1.866, 0.5) circle [radius=2pt]

(-0.5, 1.866) circle [radius=2pt]
(-1.366, 1.366) circle [radius=2pt]
(-1.866, 0.5) circle [radius=2pt]

(-0.5, -1.866) circle [radius=2pt]
(-1.366, -1.366) circle [radius=2pt]
(-1.866, -0.5) circle [radius=2pt]

(0.5, -1.866) circle [radius=2pt]
(1.366, -1.366) circle [radius=2pt]
(1.866, -0.5) circle [radius=2pt]
;

\node (up) at (-0.5, 1.866+0.25) {\small $v_1$};
\node (up) at (0.5, 1.866+0.25) {\small $v_2$};

\node (northeast) at (1.366+0.25, 1.366+0.15) {\small $v_3$};
\node (right) at (1.866+0.25, 0.5) {\small $v_4$};

\node (right) at (1.866+0.25, -0.5) {\small $v_5$};
\node (southeast) at (1.366+0.25, -1.366-0.15) {\small $v_6$};

\node (down) at (0.5, -1.866-0.25) {\small $v_7$};
\node (down) at (-0.5, -1.866-0.25) {\small $v_8$};

\node (southwest) at (-1.366-0.25, -1.366-0.15) {\small $v_9$};
\node (left) at (-1.866-0.35, -0.5) {\small $v_{10}$};

\node (left) at (-1.866-0.35, 0.5) {\small $v_{11}$};
\node (northwest) at (-1.366-0.35, 1.366+0.15) {\small $v_{12}$};

\draw[very thick] (-0.5, 1.866) to (0.5, 1.866);
\draw[very thick] (1.366, 1.366) to (0.5, 1.866);
\draw[very thick] (1.366, 1.366) to (1.866, 0.5);
\draw[very thick] (1.866, -0.5) to (1.866, 0.5);
\draw[very thick] (1.866, -0.5) to (1.366, -1.366);
\draw[very thick] (0.5, -1.866) to (1.366, -1.366);
\draw[very thick] (0.5, -1.866) to (-0.5, -1.866);
\draw[very thick] (-1.366, -1.366) to (-0.5, -1.866);
\draw[very thick] (-1.366, -1.366) to (-1.866, -0.5);
\draw[very thick] (-1.866, 0.5) to (-1.866, -0.5);
\draw[very thick] (-1.866, 0.5) to (-1.366, 1.366);
\draw[very thick] (-0.5, 1.866) to (-1.366, 1.366);

\draw[very thick] (-0.5, 1.866) to (1.366, -1.366);
\draw[very thick] (-1.366, 1.366) to (0.5, -1.866);

\draw[very thick] (0.5, 1.866) to (-1.366, -1.366);
\draw[very thick] (1.366, 1.366) to (-0.5, -1.866);

\draw[very thick] (-1.866, 0.5) to (1.866, 0.5);
\draw[very thick] (-1.866, -0.5) to (1.866, -0.5);

\end{tikzpicture}
\caption{A tight example of the $3/4$-approximation algorithm for general 4CP, where each solid edge has a weight of 1 and each omitted edge has a weight of 0}
\label{fig04}
\end{figure}
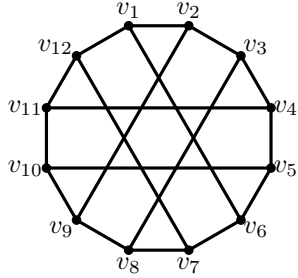

The current-best approximation ratio for general $4$PP is $3/4$~\cite{hassin1997approximation}. By the proof of Theorem~\ref{tm5}, the 4-path packing $\P_4$ holds that $w(\P_4)\geq\frac{3}{4}w(\C^*_4)$. Since $w(\C^*_4)\geq w(\P^*_4)$, we have $w(\P_4)\geq \frac{3}{4}w(\P^*_4)$. The simple algorithm of $\P_4$ achieves an approximation ratio of $3/4$. Indeed, the algorithm is the same as that in~\cite{hassin1997approximation}, where they also gave a $\{0,1,2\}$-weighted graph to show the tightness of the algorithm. We show that the analysis is tight even on $\{0,1\}$-weighted graphs. Figure~\ref{fig04+} shows an example, where there are $n=16$ vertices, the weight of each solid edge is 1, and the weight of each omitted edge is 0. The maximum weight 4-path packing $\{v_{14}v_{15}v_{16}v_1, v_2v_3v_4v_5, v_6v_7v_8v_{9},v_{10}v_{11}v_{12}v_{13}\}$ has a weight of 8. The algorithm may find a maximum weight matching $\M^*=\{v_1v_2,v_3v_4,\dots,v_{15}v_{16}\}$ of size $n/2$ with a weight of 4, and get a maximum weight 4-path packing $\P_4=\{v_1v_2v_3v_4, v_5v_6v_7v_8, v_9v_{10}v_{11}v_{12},v_{13}v_{14}v_{15}v_{16}\}$ containing the edges of $\M^*$. Then, the weight of $\P_4$ is $6$, and the approximation ratio is $6/8=3/4$.

\begin{figure}[ht]
\centering
\begin{tikzpicture}[scale=0.9]
\filldraw[black]
(0, 0.25) circle [radius=2pt]
(1, 0.75) circle [radius=2pt]
(1.5, 0.75) circle [radius=2pt]
(2.5, 0.75) circle [radius=2pt]
(3, 0.75) circle [radius=2pt]
(4, 0.75) circle [radius=2pt]
(4.5, 0.75) circle [radius=2pt]
(5.5, 0.25) circle [radius=2pt]

(0, -0.25) circle [radius=2pt]
(1, -0.75) circle [radius=2pt]
(1.5, -0.75) circle [radius=2pt]
(2.5, -0.75) circle [radius=2pt]
(3, -0.75) circle [radius=2pt]
(4, -0.75) circle [radius=2pt]
(4.5, -0.75) circle [radius=2pt]
(5.5, -0.25) circle [radius=2pt];

\node (up) at (0, 0.25+0.25) {\small $v_{16}$};
\node (up) at (1, 0.75+0.25) {\small $v_1$};
\node (up) at (1.5, 0.75+0.25) {\small $v_2$};
\node (up) at (2.5, 0.75+0.25) {\small $v_3$};
\node (up) at (3, 0.75+0.25) {\small $v_4$};
\node (up) at (4, 0.75+0.25) {\small $v_5$};
\node (up) at (4.5, 0.75+0.25) {\small $v_6$};
\node (up) at (5.5, 0.25+0.25) {\small $v_7$};

\node (up) at (0, -0.25-0.25) {\small $v_{15}$};
\node (up) at (1, -0.75-0.25) {\small $v_{14}$};
\node (up) at (1.5, -0.75-0.25) {\small $v_{13}$};
\node (up) at (2.5, -0.75-0.25) {\small $v_{12}$};
\node (up) at (3, -0.75-0.25) {\small $v_{11}$};
\node (up) at (4, -0.75-0.25) {\small $v_{10}$};
\node (up) at (4.5, -0.75-0.25) {\small $v_{9}$};
\node (up) at (5.5, -0.25-0.25) {\small $v_8$};

\draw[very thick] (0, -0.25) to (0, 0.25);
\draw[very thick] (0, 0.25) to (1, 0.75);
\draw[very thick] (2.5, 0.75) to (3, 0.75);
\draw[very thick] (1.5, 0.75) to (2.5, 0.75);
\draw[very thick] (5.5, -0.25) to (5.5, 0.25);
\draw[very thick] (5.5, -0.25) to (4.5, -0.75);
\draw[very thick] (2.5, -0.75) to (3, -0.75);
\draw[very thick] (3, -0.75) to (4, -0.75);
\end{tikzpicture}
\caption{A tight example of the $3/4$-approximation algorithm for general 4PP, where each solid edge has a weight of 1 and each omitted edge has a weight of 0}
\label{fig04+}
\end{figure}
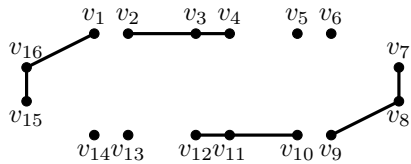

\subsection{Metric 4CP}\label{sec5.2}
Li and Yu~\cite{li2023cyclepack} proved an almost trivial approximation ratio of $3/4$. We show that their algorithm, denoted by Alg.7, actually achieves an approximation ratio of $5/6$. Moreover, we will give a tight example for this algorithm. 

\medskip
\noindent\textbf{Step~1.} Find a maximum weight matching $\M^*$ of size $n/2$.

\noindent\textbf{Step~2.} Construct a multi-graph $G/\M^*$ such that there are $n/2$ super-vertices one-to-one corresponding to the $n/2$ edges in $\M^*$, i.e., there is a function $f$, and for two super-vertices $f(e_i),f(e_j)$ such that $e_i,e_j\in\M^*$, there are two super-edges $f(e_i)f(e_j)$ between them, corresponding to the edge sets $\{uz,xy\}$ and $\{uy,xz\}$ with a weight of $w(u,z)+w(x,y)$ and $w(u,y)+w(x,z)$ ($ux\in e_i, yz\in e_j$).

\noindent\textbf{Step~3.} Find a maximum weight matching $\M^{**}_{n/4}$ of size $n/4$ in graph $G/\M^*$. Note that $\M^*\cup\M^{**}_{n/4}$ corresponds to a $4$-cycle packing $\C_4$ in graph $G$ if we decompose each super-edge of $\M^{**}_{n/4}$ into two normal edges.

\noindent\textbf{Step~4.} Return $\C_4$.
\medskip

Note that $\C_4$ is the maximum weight 4-cycle packing containing the edges of $\M^*$ by the optimality of $\M^{**}_{n/4}$. Recall that we can get a 4-path packing $\P_4$ such that $w(\P_4)\geq \frac{1}{2}w(\M^*)+\frac{1}{2}w(\C^*_4)$ by Lemma~\ref{useful}. Moreover, if $\P_4=\{u_ix_iy_iz_i\}_{i=1}^{n/4}$, $\M^*$ represents the edge set $\{u_ix_i,y_iz_i\}_{i=1}^{n/4}$. Let $\overline{\P_4}$ denote the edge set $\{u_iz_i\}_{i=1}^{n/4}$. We have the following bounds.

\begin{lemma}\label{lb6}
$w(\C_4)\geq \frac{3}{4}w(\C^*_4)+w(\overline{\P_4})$.
\end{lemma}
\begin{proof}
Obtain a $4$-cycle packing $\C'_4$ such that $C'_i=u_ix_iy_iz_iu_i\in\C'_4$ for every 4-path $u_ix_iy_iz_i\in\P_4$. Then, we can get $w(\C'_4)=w(\P_4)+w(\overline{\P_4})$. Since $\C'_4$ contains the edges of $\M^*$ and $\C_4$ is the maximum weight 4-cycle packing containing the edges of $\M^*$, we have $w(\C_4)\geq w(\C'_4)$. 
Recall that $w(\P_4)\geq \frac{1}{2}w(\M^*)+\frac{1}{2}w(\C^*_4)$. Moreover, by the proof of Theorem~\ref{tm5}, we have $w(\P_4)\geq \frac{3}{4}w(\C^*_4)$. So, we have
\[
w(\C_4)\geq w(\C'_4)\geq \frac{3}{4}w(\C^*_4)+w(\overline{\P_4}).
\]
\end{proof}

\begin{lemma}\label{lb7}
$w(\C_4)\geq w(\C^*_4)-2w(\overline{\P_4})$.
\end{lemma}
\begin{proof}
Obtain a $4$-cycle packing $\C''_4$ such that $C''_i=u_ix_iz_iy_iu_i\in\C''_4$ for every 4-path $u_ix_iy_iz_i\in\P_4$. Then, we can get 
\[
w(C''_i)=w(u_i,x_i)+w(x_i,z_i)+w(z_i,y_i)+w(y_i,u_i)\geq 2(w(u_i,x_i)+w(z_i,y_i))-2w(u_i,z_i)
\]
since $w(x_i,z_i)+w(y_i,u_i)+2w(u_i,z_i)\geq w(u_i,x_i)+w(z_i,y_i)$ by the triangle inequality. Recall that $\M^*$ represents the edge set $\{u_ix_i,y_iz_i\}_{i=1}^{n/4}$ and $\overline{\P_4}$ is the edge set $\{u_iz_i\}_{i=1}^{n/4}$. Then, we can get $w(\C''_4)\geq 2w(\M^*)-2w(\overline{\P_4})$. Since $\C''_4$ contains the edges of $\M^*$ and $\C_4$ is the maximum weight 4-cycle packing containing the edges of $\M^*$, we have $w(\C_4)\geq w(\C''_4)$. Recall that $w(\M^*)\geq \frac{1}{2}w(\C^*_4)$ by Lemma~\ref{lb3}. So, we have 
\[
w(\C_4)\geq w(\C''_4)\geq w(\C^*_4)-2w(\overline{\P_4}).
\]
\end{proof}

\begin{theorem}\label{tm6}
For metric 4CP, Alg.7 is a polynomial-time $5/6$-approximation algorithm.
\end{theorem}
\begin{proof}
By Lemmas~\ref{lb6} and \ref{lb7}, we can get 
\begin{align*}
3w(\C_4)&\geq 2\cdot\lrA{\frac{3}{4}w(\C^*_4)+w(\overline{\P_4})}+1\cdot\lrA{w(\C^*_4)-2w(\overline{\P_4})}\\
&=\frac{3}{2}w(\C^*_4)+w(\C^*_4)=\frac{5}{2}w(\C^*_4).    
\end{align*}
Therefore, we have $w(\C_4)\geq\frac{5}{6}w(\C^*_4)$, and the approximation ratio is $5/6$.
\end{proof}

Figure~\ref{fig05} shows a tight example, where there are $n=8$ vertices, each gray edge has a weight of $1$, each red edge has a weight of $2$, each blue edge has a weight of $3$, and each orange edge has a weight of $4$. One can easily verify that the graph is a metric graph. Note that the maximum weight 4-cycle packing $\{v_2v_3v_5v_4v_2,v_6v_7v_1v_8v_6\}$ has a weight of 24. The algorithm may find a maximum weight matching $\M^*=\{v_1v_2,v_3v_4,v_5v_6,v_7v_8\}$ of size $n/2$ with a weight of 12, and get a maximum weight 4-cycle packing $\C_4=\{v_1v_2v_5v_6v_1, v_3v_4v_7v_8v_3\}$ containing the edges of $\M^*$. Then, the weight of $\C_4$ is $20$, and the approximation ratio is $20/24=5/6$.

\begin{figure}[ht]
\centering
\begin{tikzpicture}[scale=0.85]
\filldraw[black]
(-0.5,1.5) circle [radius=2pt]
(-1.5,0.5) circle [radius=2pt]
(-0.5,-1.5) circle [radius=2pt]
(-1.5,-0.5) circle [radius=2pt]
(0.5,1.5) circle [radius=2pt]
(1.5,0.5) circle [radius=2pt]
(0.5,-1.5) circle [radius=2pt]
(1.5,-0.5) circle [radius=2pt];

\node (up) at (-0.5,1.75) {\small $v_{1}$};

\node (up) at (0.5,1.75) {\small $v_{2}$};
\node (down) at (-0.5,-1.75) {\small $v_{6}$};
\node (down) at (0.5,-1.75) {\small $v_{5}$};
\node (left) at (-1.75,0.5) {\small $v_{8}$};
\node (left) at (-1.75,-0.5) {\small $v_{7}$};
\node (right) at (1.75,0.5) {\small $v_{3}$};
\node (right) at (1.75,-0.5) {\small $v_{4}$};

\node[orange] (up) at (0,1.75) {\small $4$};
\node[orange] (up) at (0,-1.75) {\small $4$};
\node[red] (left) at (1.75,0) {\small $2$};
\node[red] (left) at (-1.75,0) {\small $2$};
\node[gray] at (0,0.75) {\small $1$};
\node[gray] at (0,-0.75) {\small $1$};

\node[blue] (northeast) at (1.176,1.176) {\small $3$};
\node[blue] (southeast) at (1.176,-1.176) {\small $3$};
\node[blue] (southwest) at (-1.176,-1.176) {\small $3$};
\node[blue] (northwest) at (-1.176,1.176) {\small $3$};

\draw[very thick,orange] (-0.5,1.5) to (0.5,1.5);
\draw[very thick,orange] (-0.5,-1.5) to (0.5,-1.5);

\draw[very thick,blue] (1.5,0.5) to (0.5,1.5);
\draw[very thick,blue] (1.5,-0.5) to (0.5,-1.5);
\draw[very thick,blue] (-1.5,0.5) to (-0.5,1.5);
\draw[very thick,blue] (-1.5,-0.5) to (-0.5,-1.5);

\draw[very thick,gray] (-0.5,1.5) to (1.5,0.5);
\draw[very thick,gray] (-0.5,1.5) to (1.5,-0.5);
\draw[very thick,blue] (-1.5,-0.5) to (-0.5,1.5);
\draw[very thick,gray] (0.5,1.5) to (-1.5,0.5);
\draw[very thick,gray] (0.5,1.5) to (-1.5,-0.5);
\draw[very thick,blue] (1.5,-0.5) to (0.5,1.5);

\draw[very thick,gray] (-0.5,-1.5) to (1.5,0.5);
\draw[very thick,gray] (-0.5,-1.5) to (1.5,-0.5);
\draw[very thick,blue] (-0.5,-1.5) to (-1.5,0.5);
\draw[very thick,gray] (0.5,-1.5) to (-1.5,0.5);
\draw[very thick,gray] (0.5,-1.5) to (-1.5,-0.5);
\draw[very thick,blue] (0.5,-1.5) to (1.5,0.5);

\draw[very thick,red] (-1.5,0.5) to (-1.5,-0.5);
\draw[very thick,red] (1.5,0.5) to (1.5,-0.5);
\draw[very thick,red] (-0.5,1.5) to (-0.5,-1.5);
\draw[very thick,red] (-0.5,1.5) to (0.5,-1.5);
\draw[very thick,red] (0.5,1.5) to (-0.5,-1.5);
\draw[very thick,red] (0.5,1.5) to (0.5,-1.5);
\draw[very thick,red] (-1.5,0.5) to (1.5,0.5);
\draw[very thick,red] (-1.5,0.5) to (1.5,-0.5);
\draw[very thick,red] (-1.5,-0.5) to (1.5,0.5);
\draw[very thick,red] (-1.5,-0.5) to (1.5,-0.5);
\end{tikzpicture}
\caption{A tight example of the $5/6$-approximation algorithm for metric 4CP, where each gray edge has a weight of $1$, each red edge has a weight of $2$, each blue edge has a weight of $3$, and each orange edge has a weight of $4$}
\label{fig05}
\end{figure}
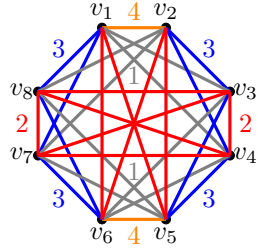

On $\{1,2\}$-weighted graphs we may obtain a better approximation ratio. 

\begin{theorem}\label{tm6+}
For metric 4CP on $\{1,2\}$-weighted graphs, Alg.7 is a polynomial-time $7/8$-approximation algorithm.
\end{theorem}
\begin{proof}
By Lemma~\ref{lb6}, we have $w(\C_4)\geq \frac{3}{4}w(\C^*_4)+w(\overline{\P_4})$. Note that $\overline{\P_4}$ is a matching of size $n/4$. On $\{1,2\}$-weighted graphs, we have $w(\overline{\P_4})\geq n/4$ and $w(\C^*_4)\leq 2n$ since every edge in $\overline{\P_4}$ has a weight of at least 1 and every edge in $\C^*_4$ has a weight of at most 2. So, $w(\overline{\P_4})\geq \frac{1}{8}w(\C^*_4)$, and then $w(\C_4)\geq\frac{7}{8}w(\C^*_4)$. The approximation ratio is $7/8$.
\end{proof}

Note that by modifying the example in Figure~\ref{fig04} such that each solid edge has a weight of 2 and each omitted edge has a weight of 1 we can prove that the analysis is tight by a similar argument.

\subsection{Metric 4PP}\label{sec5.3}
At last, we will consider metric 4PP. Recall that we can get a 4-path packing $\P_4$ such that $w(\P_4)\geq \frac{1}{2}w(\M^*)+\frac{1}{2}w(\C^*_4)$ by Lemma~\ref{useful}. For metric 4PP, we will construct another 4-path packing $\P'_4$. The algorithm, denoted by Alg.8, is shown as follows.

\medskip
\noindent\textbf{Step~1.} Obtain a 4-path packing $\P_4$ such that $w(\P_4)\geq \frac{1}{2}w(\M^*)+\frac{1}{2}w(\C^*_4)$ using Alg.6.

\noindent\textbf{Step~2.} Obtain a maximum weight matching $\M^{**}_{n/4}$ of size $n/4$ in graph $G$. Note that there are also $n/2$ isolated vertices not covered by $\M^{**}_{n/4}$.

\noindent\textbf{Step~3.} Arbitrarily assign two isolated vertices $u_i,z_i$ for each edge $x_iy_i\in\M^{**}_{n/4}$. Assume w.l.o.g. that $w(u_i,x_i)+w(y_i,z_i)\geq w(z_i,x_i)+w(y_i,u_i)$.

\noindent\textbf{Step~4.} Obtain another 4-path packing $\P'_4$ by taking a $4$-path $u_ix_iy_iz_i$ for every edge $x_iy_i\in\M^{**}_{n/4}$ with the two isolated vertices $u_i,z_i$.
\medskip

Let $\C_4$ be the 4-cycle packing obtained by completing the maximum weight 4-path packing $\P^*_4$, i.e., for every 4-path $P_i=u_ix_iy_iz_i\in \P_4$, obtain a 4-cycle $C_i=u_ix_iy_iz_iu_i$. Then, let $\C_4=\P^*_4\cup\overline{\P^*_4}$. Moreover, let $\C_4=\M_1\cup\M_2$ such that $\M_1$ and $\M_2$ are two matchings of size $n/2$, and $\M_1\cap\overline{\P^*_4}=\emptyset$.
Obtain another 4-cycle packing $\C'_4$ such that for every 4-path $P_i=u_ix_iy_iz_i\in \P_4$ there is a 4-cycle $C'_i=u_ix_iz_iy_iu_i$ in $\C'_4$.

We have the following bounds.

\begin{lemma}\label{lb8}
$w(\P_4)\geq \max\{\frac{1}{2}w(\M_1)+\frac{1}{2}w(\P^*_4)+\frac{1}{2}w(\overline{\P^*_4}),\frac{3}{2}w(\M_1)-w(\overline{\P^*_4})\}$.
\end{lemma}
\begin{proof}
Firstly, we show that $w(\P_4)\geq\frac{1}{2}w(\M_1)+\frac{1}{2}w(\P^*_4)+\frac{1}{2}w(\overline{\P^*_4})$.
Recall that $w(\P_4)\geq \frac{1}{2}w(\M^*)+\frac{1}{2}w(\C^*_4)$, $w(\C^*_4)\geq w(\C_4)=w(\P^*_4)+w(\overline{\P^*_4})$, and $w(\M^*)\geq w(\M_1)$. Then, we have 
\[
w(\P_4)\geq \frac{1}{2}w(\M_1)+\frac{1}{2}w(\C_4)=\frac{1}{2}w(\M_1)+\frac{1}{2}w(\P^*_4)+\frac{1}{2}w(\overline{\P^*_4}).
\]

Then, we show that $w(\P_4)\geq\frac{3}{2}w(\M_1)-w(\overline{\P^*_4})$. Note that we can get a 4-cycle packing $\C'_4$ such that $C'_i=u_ix_iz_iy_iu_i$ for every 4-path $P_i=u_ix_iy_iz_i\in \P^*_4$. We have 
\[
w(C'_i)=w(u_i,x_i)+w(x_i,z_i)+w(z_i,y_i)+w(y_i,u_i)\geq 2(w(u_i,x_i)+w(z_i,y_i))-2w(u_i,z_i)
\]
since $w(x_i,z_i)+w(y_i,u_i)+2w(u_i,z_i)\geq w(u_i,x_i)+w(z_i,y_i)$ by the triangle inequality. Note that $\M_1$ represents the edge set $\{u_ix_i,y_iz_i\}_{i=1}^{n/4}$ and $\overline{\P^*_4}$ represents the edge set $\{x_iy_i\}_{i=1}^{n/4}$. Hence, we can get $w(\C'_4)\geq 2w(\M_1)-2w(\overline{\P^*_4})$. Recall that $w(\P_4)\geq \frac{1}{2}w(\M^*)+\frac{1}{2}w(\C^*_4)$, $w(\C^*_4)\geq w(\C'_4)$, and $w(\M^*)\geq w(\M_1)$. Then, we have 
\begin{align*}
w(\P_4)&\geq \frac{1}{2}w(\M_1)+\frac{1}{2}w(\C'_4)\\
&\geq \frac{1}{2}w(\M_1)+w(\M_1)-w(\overline{\P^*_4})=\frac{3}{2}w(\M_1)-w(\overline{\P^*_4}).    
\end{align*}
\end{proof}

\begin{lemma}\label{lb9}
$w(\P'_4)\geq 2w(\P^*_4)-2w(\M_1)$.
\end{lemma}
\begin{proof}
For every 4-path $P'_i=u_ix_iy_iz_i\in\P'_4$, we have 
\[
w(P'_i)=w(u_i,x_i)+w(x_i,y_i)+w(y_i,z_i)\geq 2w(x_i,y_i)
\]
since $w(u_i,x_i)+w(y_i,z_i)\geq w(z_i,x_i)+w(y_i,u_i)$ by Step 3 of the algorithm, and then $w(u_i,x_i)+w(y_i,z_i)\geq \frac{1}{2}(w(u_i,x_i)+w(y_i,z_i)+w(z_i,x_i)+w(y_i,u_i))\geq w(x_i,y_i)$ by the triangle inequality. Note that $\M^{**}_{n/4}$ represents the edge set $\{x_iy_i\}_{i=1}^{n/4}$. Then, we have $w(\P'_4)\geq 2w(\M^{**}_{n/4})$. Since $\C_4=\M_1\cup\M_2$ and $\M_1\cap\overline{\P^*_4}=\emptyset$, we know that $\P^*_4\setminus\M_1$ is a matching of size $n/4$. Recall that $\M^{**}_{n/4}$ is the maximum weight matching of size $n/4$. We can get $w(\M^{**}_{n/4})\geq w(\P^*_4)-w(\M_1)$. Then, we have 
\[
w(\P'_4)\geq 2w(\M^{**}_{n/4})\geq 2w(\P^*_4)-2w(\M_1).
\]
\end{proof}

\begin{theorem}\label{tm7}
For metric 4PP, there is a polynomial-time $14/17$-approximation algorithm.
\end{theorem}
\begin{proof}
The algorithm returns a better 4-path packing from $\P_4$ and $\P'_4$. By Lemmas~\ref{lb8} and \ref{lb9}, we have 
\begin{align*}
&17\cdot\max\{w(\P_4),w(\P'_4)\}\\
&\geq 12w(\P_4)+5w(\P'_4)\\
&\geq 8\cdot\lrA{\frac{1}{2}w(\M_1)+\frac{1}{2}w(\P^*_4)+\frac{1}{2}w(\overline{\P^*_4})}+4\cdot\lrA{\frac{3}{2}w(\M_1)-w(\overline{\P^*_4})}\\
&\quad\quad+5\cdot\lrA{2w(\P^*_4)-2w(\M_1)}\\
&=4w(\M_1)+4w(\P^*_4)+4w(\overline{\P^*_4})+6w(\M_1)-4w(\overline{\P^*_4})+10w(\P^*_4)-10w(\M_1)\\
&=14w(\P^*_4).
\end{align*}
Therefore, we have $\max\{w(\P_4),w(\P'_4)\}\geq\frac{14}{17}w(\P^*_4)$, and the approximation ratio is $14/17$.
\end{proof}

Note that the approximation ratio is achieved by doing a trade-off between two 4-path packings. A tight example may not be easy to construct. So, we will not consider the tightness of its analysis. We also remark that on $\{1,2\}$-weighted graphs the current-best approximation ratio is $9/10$~\cite{DBLP:journals/jda/MonnotT08}. %By a tighter analysis of our algorithm, we cannot improve it.

\section{Approximation algorithms on $\{1,2\}$-weighted graphs}\label{sec7}
In this section, we study approximation algorithms on $\{1,2\}$-weighted graphs. 
Our motivation is that there are many studies on $\{0,1\}$-weighted graphs~\cite{DBLP:journals/dam/Bar-NoyPRV18,hassin2013local,DBLP:conf/soda/BermanK06} but a few on $\{1,2\}$-weighted graphs~\cite{DBLP:journals/jda/MonnotT08}. We show that we can reduce $k$CP/PP on $\{1,2\}$-weighted graph to $k$CP/PP on $\{0,1\}$-weighted graphs with some properties related to the approximation ratios.
%Since there are many studies on $\{0,1\}$-weighted graphs~\cite{DBLP:journals/dam/Bar-NoyPRV18,hassin2013local,DBLP:conf/soda/BermanK06} but a few on $\{1,2\}$-weighted graphs~\cite{DBLP:journals/jda/MonnotT08}, it is natural to consider approximation algorithms on $\{1,2\}$-weighted graphs.
% and give a $9/11$-approximation algorithm for 3CP on $\{1,2\}$-weighted graphs.

%\subsection{A Black-Box Reduction}
%We can reduce $k$CP/$k$PP on $\{1,2\}$-weighted graphs to $k$CP/$k$PP on $\{0,1\}$-weighted graphs.

\begin{theorem}\label{spec}
For $k$CP/$k$PP, given a $\rho$-approximation algorithm on $\{0,1\}$-weighted graphs, we can obtain an $(1+\rho)/2$-approximation algorithm on $\{1,2\}$-weighted graphs.
\end{theorem}
\begin{proof}
We only consider $k$CP, and the analysis can be extended to $k$PP.

The reduction is shown as follows. Given an $\{1,2\}$-weighted graph $G$, we first obtain a $\{0,1\}$-weighted graph $G'$ by regarding each 1-weighted edge in $G$ as a 0-weighted edge and each 2-weighted edge in $G$ as an 1-weighted edge. Then, we call a $\rho$-approximation algorithm for $k$CP on $G'$ to obtain a $k$-cycle packing $\C'_k$, which corresponds to a $k$-cycle packing $\C_k$ in $G$. 
Note that 
$
w(\C_k)=w(\C'_k)+n.
$
It is easy to see that the reduction takes polynomial time. 
Next, we analyze the quality of $\C_k$.

Let $\C^{*}_k$ (resp., $\C'^{*}_k$) denote the optimal solution of $k$CP on $G$ (resp., $G'$). Similarly, we have $w(\C^{*}_k)=w(\C'^{*}_k)+n$.
Then, the number of 1-weighted edges in $\C'^{*}_k$ is $w(\C'^{*}_k)$, and the number of 0-weighted edges in $\C'^{*}_k$ is $n-w(\C'^{*}_k)$. Hence, we have 
$
n-w(\C'^{*}_k)\geq 0.
$
Since $\C'_k$ is a $\rho$-approximate solution on $G'$, we have $w(\C'_k)\geq \rho \cdot w(\C'^{*}_k)$. We can get
\begin{align*}
\frac{w(\C_k)}{w(\C^*_k)}&=\frac{w(\C'_k)+n}{w(\C'^*_k)+n}\\
&\geq\frac{\rho\cdot w(\C'^{*}_k)+n}{w(\C'^*_k)+n}=\frac{(1+\rho)\cdot w(\C'^{*}_k)+n-w(\C'^{*}_k)}{2w(\C'^*_k)+n-w(\C'^{*}_k)}\geq \frac{1+\rho}{2},    
\end{align*}
where the last inequality follows from $n-w(\C'^{*}_k)\geq 0$ and $1\geq\rho\geq 0$.

Therefore, we can get an $(1+\rho)/2$-approximate $k$-cycle packing $\C_k$ on $G$ in polynomial time.
\end{proof}

%Note that in the worst case we have $n-$
For 3CP on $\{0,1\}$-weighted graphs, there is a $3/5$-approximation algorithm~\cite{DBLP:journals/dam/Bar-NoyPRV18}. By Theorem~\ref{spec}, we can directly get a $4/5$-approximation algorithm for 3CP on $\{1,2\}$-weighted graphs. We show that based on the property of their algorithm, we can get an approximation ratio of $9/11$.

\begin{theorem}
For 3CP on $\{1,2\}$-weighted graphs, there is a $9/11$-approximation algorithm.
\end{theorem}
\begin{proof}
We use $G$ to denote an $\{1,2\}$-weighted graph, and $G'$ to denote the $\{0,1\}$-weighted graph obtained by the reduction in the proof of Theorem~\ref{spec}. Let $\C^{*}_3$ (resp., $\C'^{*}_3$) denote the optimal solution of 3CP on $G$ (resp., $G'$), and $n_i$ denote the number of $i$-weighted 3-cycles in $\C'^{*}_3$, where $i\in\{0,1,2,3\}$. Note that $n/3= n_3+n_2+n_1+n_0$ and $w(\C'^*_3)=3n_3+2n_2+n_1$.

We use the $3/5$-approximation algorithm~\cite{DBLP:journals/dam/Bar-NoyPRV18} to get a 3-cycle packing $\C'_3$ on $G'$. Bar-Noy et al.~\cite{DBLP:journals/dam/Bar-NoyPRV18} proved that
\[
w(\C'_3)\geq \max\{2n_3+n_2+n_1,\ \delta\cdot(2n_3+2n_2+n_3)\},
\]
where $\delta=3/4$ is the approximation ratio of 3PP on $\{0,1\}$-weighted graphs~\cite{DBLP:journals/dam/Bar-NoyPRV18}.

By the proof of Theorem~\ref{spec}, the obtained 3-cycle packing $\C_3$ (by the reduction) with respect to $\C'_3$ has a weight of $w(\C_3)=w(\C'_3)+n$, and it holds that $w(\C^{*}_3)=w(\C'^{*}_3)+n$. Hence, we have
\begin{align*}
w(\C_3)&=w(\C'_3)+n\\
&\geq\max\lrC{2n_3+n_2+n_1,\ \frac{3}{4}\cdot(2n_3+2n_2+n_3)}+n\\
&\geq\frac{9}{11}(2n_3+n_2+n_1)+\frac{2}{11}\lrA{\frac{3}{2}n_3+\frac{3}{2}n_2+\frac{3}{4}n_3}+n\\
&=\frac{21}{11}n_3+\frac{12}{11}n_2+\frac{21}{22}n_1+n\\
&=\frac{21}{11}n_3+\frac{12}{11}n_2+\frac{21}{22}n_1+(3n_3+3n_2+3n_1)+(n-3n_3-3n_2-3n_1)\\
&=\frac{54}{11}n_3+\frac{45}{11}n_2+\frac{87}{22}n_1+(n-3n_3-3n_2-3n_1)\\
&\geq\frac{54}{11}n_3+\frac{45}{11}n_2+\frac{36}{11}n_1+(n-3n_3-3n_2-3n_1).
\end{align*}
Recall that $w(\C^{*}_3)=w(\C'^{*}_3)+n$ and $w(\C'^*_3)=3n_3+2n_2+n_1$. We can get $w(\C^{*}_3)=3n_3+2n_2+n_1+n=6n_3+5n_2+4n_1+(n-3n_3-3n_2-3n_1)$. Then, we have
\[
\frac{w(\C_3)}{w(\C^*_3)}\geq\frac{\frac{54}{11}n_3+\frac{45}{11}n_2+\frac{36}{11}n_1+(n-3n_3-3n_2-3n_1)}{6n_3+5n_2+4n_1+(n-3n_3-3n_2-3n_1)}\geq\frac{9}{11},
\]
where the last inequality follows from $n=3n_3+3n_2+3n_1+3n_0\geq 3n_3+3n_2+3n_1$.
\end{proof}

\section{Conclusion}\label{sec6}
In this paper, we consider approximation algorithms for metric/general $k$CP and $k$PP. Most of our results are based on simple algorithms but with deep analysis.
In the future, it would be interesting to improve these approximation ratios, even on $\{0,1\}$-weighted or $\{1,2\}$-weighted graphs. One challenging direction is to design better algorithms for metric/general 3CP and 3PP.

\bibliographystyle{plain}
\bibliography{main}
\end{document}